%% file: main.tex
\documentclass[11pt,english,a4paper]{article}

\usepackage[T1]{fontenc}
\usepackage[margin=1in]{geometry}
\usepackage[utf8]{inputenc}
\usepackage{color}
\usepackage[dvipsnames]{xcolor}
\usepackage{enumerate}
\usepackage{enumitem}
\usepackage{hyperref}
\definecolor{darkblue}{rgb}{0, 0, 0.5}
\hypersetup{
     colorlinks = true,
     citecolor = darkblue,
     urlcolor = darkblue,
     linkcolor = darkblue
}
\usepackage{amsmath}
\usepackage{amsthm}
\usepackage{multirow}
\usepackage{amssymb}
\usepackage{cases}
\usepackage{bm}
\usepackage[authoryear]{natbib}
\usepackage{setspace}
\usepackage{graphicx}
\usepackage{float}
\usepackage{cleveref}
\usepackage{booktabs}
\usepackage{soul}
\usepackage{siunitx}

\newcommand{\R}{\mathbb{R}}
\newcommand{\E}{\operatorname{\mathbb{E}}}
\newcommand{\Pb}{\operatorname{\mathbb{P}}}
\newcommand{\supp}{\operatorname{supp}}

\newcommand{\M}{\mathbb{M}}
\newcommand{\bydef}{:=}

\makeatletter
\newtheorem{assumption}{\protect\assumptionname}
\newtheorem{thm}{\protect\theoremname}
\newtheorem*{thm*}{\protect\theoremname}

\newtheorem{lem}{\protect\lemmaname}
\theoremstyle{remark}

\providecommand{\assumptionname}{Assumption}
\providecommand{\claimname}{Claim}
\providecommand{\lemmaname}{Lemma}
\providecommand{\remarkname}{Remark}
\providecommand{\theoremname}{Theorem}
\providecommand{\examplename}{Example}
\providecommand{\corname}{Corollary}

\crefname{thm}{Theorem}{Theorems}
\crefname{assumption}{Assumption}{Assumptions}
\crefname{subassumptioni}{Assumption}{Assumptions}

\newlist{subassumption}{enumerate}{1}
\setlist[subassumption,1]{
  label=(\roman*),   % 1, 2, ...
  ref=\theassumption(\roman*), % 1.1, 1.2, ...
  leftmargin=2em
}

\title{\textbf{Raking for Estimation and Inference in Panel Models\\
with Nonignorable Attrition and Refreshment}\footnote{The project reported in this publication relies on data from surveys administered by the Understanding America Study (UAS), which is maintained by the Center for Economic and Social Research (CESR) at the University of Southern California (USC). This project was supported by the National Institute on Aging of the National Institutes of Health and, in part, by the Social Security Administration under Award Number U01AG077280. The content is solely the responsibility of the authors and does not necessarily represent the official views of the National Institutes of Health, USC, or UAS.}
}   

\author{
\setcounter{footnote}{1}
    \textsc{Grigory Franguridi}\thanks{
    Center for Economic and Social Research, University of Southern California.
    Email: \href{franguri@usc.edu}{franguri@usc.edu}
    }
    \hspace{0.8cm}
    \textsc{Jinyong Hahn}\thanks{
    Department of Economics, University of California, Los Angeles. Email: \href{hahn@econ.ucla.edu}{hahn@econ.ucla.edu}
	}
	\hspace{0.8cm}
    \textsc{Pierre Hoonhout}\thanks{
    ISEG Research, ISEG Lisbon School of Economics \& Management, Universidade de Lisboa, Lisbon, Portugal.
    Email: \href{pierre@iseg.ulisboa.pt}{pierre@iseg.ulisboa.pt}
	}
    \\
    \textsc{Arie Kapteyn}\thanks{
    Center for Economic and Social Research, University of Southern California.
		Email: \href{kapteyn@usc.edu}{kapteyn@usc.edu}
	}
    \hspace{0.8cm}
	\textsc{Geert Ridder}\thanks{
	Department of Economics, University of Southern California.
		Email: \href{ridder@usc.edu}{ridder@usc.edu}
	}
}

\begin{document}

\maketitle

% \begin{center}
% DRAFT. PLEASE DO NOT CIRCULATE.    
% \end{center}

\input{sections/abstract}
\input{sections/intro}
\input{sections/framework}
\input{sections/estimation}

\input{sections/inference}

\input{sections/implementation}
\input{sections/MC}
\input{sections/empirical}

\input{sections/conclusion}

\bibliographystyle{ecta}
\bibliography{references}

\input{sections/appendix}
 
\end{document}

%% file: sections/abstract.tex
\begin{abstract}

    \linespread{1.2}

    In panel data subject to nonignorable attrition, auxiliary (refreshment) sampling may restore full identification under weak assumptions on the attrition process.
    Despite their generality, these identification strategies have seen limited empirical use, largely because the implied estimation procedure requires solving a functional minimization problem for the target density.
    When the distributions of the observed data are parametric, we show that this problem can be solved using the iterative proportional fitting (raking) algorithm, which converges rapidly even with continuous data.
    The resulting density estimator is then used as input into a parametric moment condition.
    We establish consistency and convergence rates for both the raking-based density estimator and the resulting moment estimator.
    We also derive a simple recursive procedure for estimating the asymptotic variance.
    Finally, we demonstrate the satisfactory performance of our estimator in simulations and provide an empirical illustration using data from the Understanding America Study panel.

\medskip


\medskip


\end{abstract}
\newpage

%% file: sections/intro.tex
\section{Introduction}

Longitudinal data have many advantages over (repeated) cross-sections. However, a major drawback is that panel data suffer from panel attrition, in addition to the initial nonresponse common to panel and cross-sectional surveys. Panel attrition reduces the effective sample size. More concerning is that it may bias the estimates obtained from the panel survey. For instance, if we estimate the average change in household income, then households that experienced a decline in income due to the unemployment of the head of the household may be more likely to move and be lost to follow-up. This nonrandom attrition that depends on an outcome variable results in an upward bias in the estimate of the average change. 

Estimates that depend on outcomes in all panel waves, including the wave in which the subject drops out, are biased. Unbiasedness can be restored if the probability of attrition is restricted to depend on the outcome in the wave in which the subject drops out, but not on the outcomes in previous waves \citep{hausman1979attrition}. Estimates can also be unbiased if attrition depends on the outcomes in the waves prior to the wave in which the subject drops out, but not on the (unobserved) outcome in the wave in which the subject leaves the panel \citep{rubin1976inference,little2019statistical}. The latter case is referred to as Missing At Random (MAR), with MCAR (Missing Completely At Random) being the special case in which the attrition is purely random.

The unselected joint distribution of the outcome variables (i.e., the joint distribution of the outcome variables over the two periods if there were no attrition problem) is, under the assumption of MAR, nonparametrically just identified. Therefore, additional information is needed to relax MAR if we want to allow the attrition to depend on the unobserved outcome in the wave in which the subject drops out. A source of additional information is a supplementary sample drawn to compensate for the loss of subjects due to dropout. The idea to add such a sample dates back at least to \cite{kish1959replacement}. In panel surveys, these samples are drawn in the second and later waves and are called refreshment samples \citep{ridder1992empirical}. \cite{deng2013handling} and \cite{watson2021refreshment} survey the use of refreshment samples in panel surveys.

A key identification result in panel surveys with selective attrition and refreshment samples was established by \cite{hirano2001combining}. They considered a selective panel with two waves supplemented by a refreshment sample in the second wave. In this setup, the marginal distributions of the first-wave variables and of the second-wave variables are directly identified. They proposed estimating the unselected joint distribution by the distribution that minimizes the distance to the observed (selective) joint distribution under the constraints on the marginal distributions.

Because the distance measure is a convex functional and the restrictions are linear, this functional minimization problem has a unique solution. This solution can be used to obtain a semiparametric estimation procedure. Consider an estimator that depends on the outcomes in both waves of the panel. An example is a fixed-effects (FE) regression where we regress the change in the dependent variable on the change in the independent variables. If there is no selective attrition that depends on the first- and second-wave dependent variables, then the FE estimator of the regression coefficients is unbiased. The estimator is the solution to a sample moment condition. If the attrition does depend on the first- and second-wave dependent variables, then the FE estimator is biased. 

The functional minimization problem is characterized by a first-order condition that has a unique solution. This solution (for the unselected joint distribution) depends on the (observed) distributions in both waves of the panel. The probability of observation is assumed to be additively separable in the outcomes for the two waves, which is intuitive because we only observe the selected joint distribution. Under this restriction, the probability is nonparametrically identified. 

The probability of observation is used in a weighted generalized moment estimator. \cite{bhattacharya2008inference} studies the properties of this semiparametric estimator, and \cite{hoonhout2019nonignorable} develop the estimator for panel data with three or more waves. \cite{hirano1998combining} and \cite{deng2013handling} choose parametric models for the probability of observation and the binary outcome variables and use Markov Chain Monte Carlo to compute posteriors for the parameters. \cite{deng2013handling} use multiple imputation inference \citep{rubin2004multiple}.
\citet{zheng2025semiparametric} suggest using a parametric specification for the nonignorable missingness mechanism.
\citet{sadinle2019sequentially} propose a Bayesian approach that relies on a parametric likelihood.
A disadvantage of the weighted moment and Bayesian estimators is that they are computationally demanding and require close monitoring by the statistician. \cite{bhattacharya2008inference} reports that the moment estimator does not always converge.\footnote{See \cite{Franguridi2024} for further comments on \cite{bhattacharya2008inference}.}
Preferably, we would like a procedure that does not require close monitoring.
Moreover, all existing approaches are either computationally demanding or use Bayesian or likelihood methods that require distributional assumptions.

In this paper, we follow the generalized moment approach, but instead of estimating weights, we average over the unselected joint distribution that is estimated by minimizing the distance to the observed joint distribution under the constraints on the marginal unselected distributions, a problem that has a unique solution.
If the distance is the Kullback-Leibler divergence, there is an efficient algorithm for this constrained minimization problem, which we apply to the case where the attrition process is characterized by an exponential link function.
% We exploit the availability of an efficient algorithm for this constrained minimization problem if the distance is the Kullback-Leibler divergence. This algorithm is iterative and involves only arithmetic operations.
This algorithm is known as \emph{raking} or \emph{iterative proportional fitting} and was first proposed by \cite{deming1940least} for the case of two discrete variables.

Starting from the observed joint distribution of the outcomes, often only a few steps are required for convergence to the unique solution.\footnote{\cite{ruschendorf1995convergence} gave conditions under which the algorithm converges.}
Inputs in the algorithm are the observed joint distribution of the variables in the two waves and the marginal distributions in the first and second waves, with the latter obtained from the refreshment sample. The output is an estimate of the unselected joint distribution of the variables in the two waves. We can estimate the input densities either parametrically or nonparametrically.  In this paper, we take the parametric approach. As a by-product, the raking algorithm yields a recursive approximation of the Jacobian and the asymptotic variance of the generalized moment estimator.

The paper is organized as follows.
\Cref{sec:framework} introduces the framework.
\Cref{sec:estimation} considers estimation.
\Cref{sec:inference} discusses inference.
\Cref{sec:numerical} provides details for numerical implementation.
\Cref{sec:mc} reports a simulation study of our estimator.
\Cref{sec:empirical} contains the empirical illustration.
\Cref{sec:conclusion} concludes.

%% file: sections/framework.tex
\section{Framework}\label{sec:framework}

We consider a two-period panel data model with attrition and refreshment as in \citet{hirano2001combining}.
Let $Z_{it}$ be a $d$-dimensional vector of all the time-varying variables (both outcomes and covariates) for unit $i$ in period $t \in \{1,2\}$, and let $X_i$ be a vector of all the time-invariant variables for unit $i$.
We often drop the subscript $i$ for convenience.
Our analysis can be made conditional on $X_i$,  and hence, throughout the rest of the paper, we drop the variable $X_i$.
In the first period, all the $n$ sampled units are observed: $Z_{i1},$ $i=1,\dots,n$, is a random sample from the marginal distribution of $Z_1$ with density $f_1$.
In the second period, there is \emph{attrition}: only the units with $W_i=1$ stay in the sample, while the remaining units drop out.
To compensate for this attrition, a \emph{refreshment sample} $Z_{i2}^r$, $i =1,\dots, n_r$, i.e., a random sample from the marginal distribution of $Z_2$ with density $f_2$, is drawn in period 2.
We assume that the refreshment sample is independent of the rest of the data, and its size is of the same order of magnitude as the size of the first-period sample, i.e., $n/n_r \to c\in(0,\infty)$.

Our goal is to estimate and perform inference on the $d_\theta$-dimensional parameter $\theta$ satisfying the moment condition
\begin{align}
\E_f\left[  \varphi\left(  Z_{1},Z_{2},\theta\right)  \right] = \iint \varphi\left(  z_{1},z_{2},\theta\right) f(z_1,z_2) \, d\mu(z_1,z_2) = 0, \label{eq:moment-theta-unselect}
\end{align}
where $\varphi$ is a $d_{\theta}$-dimensional moment function, and the expectation is taken with respect to the \emph{unconditional} joint density $f$ of $(Z_1,Z_2)$ relative to some dominating probability measure $\mu$ on $\R^{2d}$. This formulation allows us to handle both continuous and discrete variables simultaneously.\footnote{We consider the case where $\theta$ is just identified. If the parameter is overidentified, then (\ref{eq:moment-theta-unselect}) is the first-order condition of a quadratic minimization problem.} 

Without further restrictions, $f$ is not point identified, and hence neither is $\theta$.
Moreover, the bounds on the identified set are expected to be uninformative: any distribution $f$ with marginals $f_1$ and $f_2$ is consistent with the data.
To understand the intuition, consider the identity for the target distribution with the balanced panel distribution $f^w(z_1,z_2)\bydef f(z_1,z_2|W=1)$ and the probability of observation
\begin{align}
f(z_1,z_2) = \frac{\Pb(W=1)}{\Pb(W=1|Z_1=z_1,Z_2=z_2)} \cdot f^w(z_1,z_2). \label{eq:key-identity}
\end{align}
Consider the case of discrete data for simplicity.
The only information available to identify the weights $\Pb(W=1)/\Pb(W=1|Z_1=z_1,Z_2=z_2)$, which generally have $|\supp Z_1|\cdot|\supp Z_2|$ degrees of freedom, is the marginal distributions $f_1$ and $f_2$, which only have $|\supp Z_1|+|\supp Z_2|$ degrees of freedom.\footnote{Here, $\supp Z_1$ denotes the support of $Z_1$, and $|\supp Z_1 |$ denotes the number of support points of $Z_1$.}
This establishes the lack of identification and suggests imposing restrictions on the weights.

\citet{hirano2001combining} impose the following separability condition on the weights.\footnote{\citet{franguridi2024closed} explore an alternative separability assumption under which the model is point identified. \citet{franguridi2025inference} drop assumptions on attrition completely, and \citet{franguridi2025generalized} further consider the case where refreshment samples are not available. In these two cases, the model becomes partially identified.}
\begin{assumption}[additive nonignorability]\label{as:AN}
For a known, increasing function $G$ and unknown functions $k_1,k_2$,
\begin{align*}
    \Pb(W=1|Z_1=z_1,Z_2=z_2) = G\left( k_1(z_1)+k_2(z_2) \right).
\end{align*}    
\end{assumption}

Throughout this paper, we further require the following.
\begin{assumption}[exponential link function]\label{as:AN-exp}
$G(x)=\exp(x)$, $x\in \R$.
\end{assumption} 
This link function corresponds to adopting the Kullback-Leibler divergence as the distance mentioned above.
This assumption ensures that the functional minimization problem introduced below can be solved by raking.
% Because the cross derivative of $\exp(k_{1}(z_{1}) + k_{2}(z_{2})$ with respect to $z_{1}$ and $z_{2}$ is not equal to zero in general, we can see that the assumption allows for cross-period interactions in attrition, and as such, the exponential link function does not constitute a significant loss of generality.
On the other hand, it implies that our results do not cover non-exponential attrition link functions.
There is also an implicit restriction $k_1(z_1)+k_2(z_2)\le 0$ because the exponential function would not lie in $[0,1]$ otherwise.
The raking algorithm does not use this restriction.

From now on, we refer to \cref{as:AN,as:AN-exp} as the additively nonignorable (AN) model.
In this model, the key identity \eqref{eq:key-identity} becomes
\begin{align}
    f(z_1,z_2) = \exp\left(\tilde k_1(z_1)+\tilde k_2(z_2)\right) \cdot f^w(z_1,z_2), \label{eq:key-identity-exp}
\end{align}
where $\tilde k_1,\tilde k_2$ are defined to absorb the constant $\Pb(W=1)$ and switch the sign on $k_1,k_2$, i.e.,
\begin{align*}
    \tilde k_1(z_1)=-k_1(z_1) + \log \Pb(W=1), \quad \tilde k_2(z_2)=-k_2(z_2).
\end{align*}
The marginal densities $f_1,f_2$ can now be used to identify $\tilde k_1, \tilde k_2$,  but it is convenient to estimate $f$ using the characterization in \citet{hirano2001combining}.

%% file: sections/estimation.tex
\section{Estimation}\label{sec:estimation}

\citet{hirano2001combining} show that Assumptions \ref{as:AN} and \ref{as:AN-exp} imply that $f$ is the Kullback-Leibler projection of $f^w$ on the set $\Pi(f_1,f_2)$ of distributions with marginals $f_1$ and $f_2$, i.e.,\footnote{This projection is also known as the \emph{Schr\"{o}dinger bridge} with marginals $f_1$, $f_2$ and the reference distribution $f^w$. It is named after the physicist Erwin Schr\"{o}dinger, who first considered it in \citet{schrodinger1931umkehrung}.}
\begin{align}
    f = \arg\min_{\tilde f \in \Pi(f_1,f_2)} \operatorname{KL}(\tilde f\,||\, f^w) \label{eq:KL-representation},
\end{align}
where
\[
\operatorname{KL}(\tilde f\,||\, f^w) \bydef \iint \tilde f(z_1,z_2) \log \left( \frac{\tilde f(z_1,z_2)}{f^w(z_1,z_2)} \right) \, d\mu(z_1,z_2).
\]
Equivalently,  $f$ minimizes the Kullback-Leibler distance to $f^w$ for densities $f$ with given marginals $f_1,f_2$. Because the distance $\operatorname{KL}(\tilde f, f^w)$ viewed as a function of $\tilde f$ is strictly convex in $\tilde f$, the minimum is unique. The density $f$ is a functional of the densities $f_1,f_2$, and $f^w$, which are directly estimable from the data.

The sample analog estimator of $f$ is
\begin{align}
    \hat f = \arg\min_{\tilde f \in \Pi(\hat f_1,\hat f_2)} \operatorname{KL}(\tilde f\,||\, \hat f^w) \label{eq:KL-representation-1},
\end{align}
where $\hat f_1$, $\hat f_2$, and $\hat f^w$ are density estimators.
The estimator $\hat f$ is then a functional of $\hat f_1, \hat f_2$, and $\hat f^w$.

At first sight, solving the constrained functional optimization problem \eqref{eq:KL-representation} that defines the estimator seems computationally infeasible. However, there exists an algorithm for solving this problem that is guaranteed to converge and requires little beyond basic arithmetic operations.
To describe the algorithm, let $\Pi_1$ and $\Pi_2$ be the operators calculating the first and second marginals, respectively, i.e. $\Pi_1 f(z_1) \bydef \int f(z_1,z_2) \, d \mu_2(z_2) $ and $\Pi_2 f(z_2) \bydef \int f(z_1,z_2) \, d \mu_1(z_1)$.
Given the marginal densities $f_1,f_2$, define the operators $\phi_{1,f_1}$, $\phi_{2,f_2}$, and $\phi_{f_1,f_2}$ by
\begin{align*}
\phi_{1,f_1}(f) \bydef \frac{f_1 \cdot f}{\Pi_1 f},
\quad
\phi_{2,f_2}(f) \bydef \frac{f_2 \cdot f}{\Pi_2 f},
\quad
\phi_{f_1,f_2} \bydef \phi_{2,f_2}(\phi_{1,f_1}(f)).
\end{align*}
Hence, $\phi_{1,f_1}(f)$ is a joint pdf with the first marginal and the conditional pdf
\[
\Pi_1 \left[\phi_{1,f_1}(f)\right](z_1)=f_1(z_1), \quad \frac{\phi_{1,f_1}(f)(z_1,z_2)}{f_1(z_1)}=\frac{f(z_1,z_2)}{f_1(z_1)}.
\]
We show that $\phi_{1,f_1}(f)$ is the Kullback-Leibler (KL) projection of $f$ on the set of distributions with the first marginal $f_1$, and similarly for $\phi_{2,f_2}(f)$ \citep[p. 1164]{ruschendorf1995convergence}.
To see this, consider the KL distance between $g$ and $f$, where we substitute
\begin{align*}
    \frac{g(z_1,z_2)}{f(z_1,z_2)} = \frac{g(z_1) g(z_2|z_1)}{f(z_1)f(z_2|z_1)},
\end{align*}
so that 
\begin{align*}
    \operatorname{KL}(g\,||\,f) &= \int g(z_1,z_2)\log \frac{g(z_1,z_2)}{f(z_1,z_2)}\,dz_1 dz_2 \\
    &= \iint g(z_1,z_2) \log \frac{g(z_1)}{f(z_1)}\, dz_1 dz_2 + \iint g(z_1,z_2) \log \frac{g(z_2|z_1)}{f(z_2|z_1)}\, dz_1 dz_2 \\
    &= \int g(z_1)\log \frac{g(z_1)}{f(z_1)}\, dz_1 + \int \left[\int g(z_2|z_1) \log \frac{g(z_2|z_1)}{f(z_2|z_1)} \, dz_2 \right] g(z_1) dz_1 \\
    &= \operatorname{KL}(g_1\,||\, f_1) + \int \operatorname{KL}(g(\cdot|z_1) \,||\, f(\cdot|z_1)) g(z_1) \, dz_1.
\end{align*}
\noindent To obtain the projection we minimize this expression over $g$, where we restrict $g$ to have the first marginal $f_1$. We note that the first term is constant over this class, and hence is irrelevant.
On the other hand, the second term is nonnegative and is minimized at $g(\cdot|z_1)=f(\cdot|z_1)$ for all $z_1$ in the support of $Z_1$. Similarly, $\phi_{2,f_2}(f)$ is the KL projection on the set of densities with the second marginal density $f_2$. Hence $\phi_{2,f_2}$ projects $f$ on the set of joint densities with marginal $f_2$, leaving the conditional density $f(\cdot|z_2)$ unchanged. Further discussion can be found in \citet[paragraph 2, p. 155]{csiszar1975divergence}, \citet{ireland1968contingency}, and \citet{kullback1968probability}.

The algorithm requires assumptions on the supports of the involved distributions.
\begin{assumption}\label{as:abs-cont}
\begin{subassumption}
    \item $\mu = \bigotimes_{k=1}^{2d} \mu_k$, where $\mu_k$ is either the Lebesgue measure on $\R$ or the counting measure on a countable subset of $\R$, and the joint distribution of $Z_1,Z_2$ is absolutely continuous w.r.t. $\mu$.
    \item The support of the joint distribution of $Z_1,Z_2$ in the balanced panel is contained in the product set of the supports of $Z_1$ and $Z_2$ in the balanced panel.
    \item The support of the unselected distribution of $Z_1$ is contained in the support of $Z_1$ in the balanced panel and similarly for $Z_2$.
    \item There exists $c>0$ such that $\frac{f^w(z_1,z_2) f_2^w(z_2)}{f_2(z_2)} \ge c$ for $f^w$-a.e. $(z_1,z_2)$.
\end{subassumption}
\end{assumption}
\cref{as:abs-cont}(i) states that the variables in the data are either discrete or continuous.
\cref{as:abs-cont}(ii) prevents $f^w$ from concentrating on lower-dimensional subsets of $\R^{2d}$; as long as the latter does not happen, the condition is satisfied.
\cref{as:abs-cont}(iii) posits that the support of the period marginals of the balanced panel distribution cannot be smaller than the support of the corresponding marginals of the target distribution.
Finally, \cref{as:abs-cont}(iv) is satisfied if $f^w$ is bounded away from zero on its support.

Theorem 3.5 in \citet{ruschendorf1995convergence} implies the following result.

\begin{thm}\label{thm:raking}
    Under \cref{as:AN,as:AN-exp,as:abs-cont}, the sequence of functions 
    \begin{align*}
    \hat f^{(t)} = \hat \phi_{f_1,f_2}\left (\hat f^{(t-1)}\right ),  \quad \hat f^{(0)}= \hat f^w,    
    \end{align*}
    converges to the solution $\hat f$ of \eqref{eq:KL-representation-1} in $L^1$ as $t \to \infty$.
\end{thm}
\begin{proof}
    See \Cref{app:thm1}.
\end{proof}

The iterative equation of the theorem is the general data version of \emph{raking}, also known as \emph{iterative proportional fitting} or \emph{Sinkhorn's algorithm}.\footnote{
Raking was introduced by \citet{deming1940least} to adjust contingency tables to given marginals and was generalized to the continuous case by \citet{ireland1968contingency} and \citet{kullback1968probability}.}
It is widely used in survey statistics and, since the seminal contribution of \citet{cuturi2013sinkhorn}, has attracted substantial attention in machine learning due to its application to solving optimal transport problems with entropic regularization.

\cref{thm:raking} suggests the following two-step estimation strategy.

\begin{enumerate}
    \item Use (\ref{eq:KL-representation}) with estimated densities  $\hat f_1$, $\hat f_2$, and $\hat f^w$ to obtain an estimator $\hat f$ computed using raking as in \cref{thm:raking}:
    \begin{align*}
        \hat f = \phi_{\hat f_1,\hat f_2}^{(T)}(\hat f^w),
    \end{align*}
    where the number of iterations $T$ is determined by a stopping rule, e.g., the maximum difference between subsequent iterations being smaller than some tolerance level.
    
    \item Compute $\hat \theta$ as a solution of the moment restriction
    \begin{align}
        \iint \varphi(z_1,z_2,\hat\theta) \hat f(z_1,z_2) \, d\mu(z_1,z_2) = 0. \label{eq:thetahat}
    \end{align}   
\end{enumerate}

We now show consistency of $\hat f$ and $\sqrt n$-consistency of $\hat \theta$. 
For simplicity, in what follows we assume that $\mu$ is the Lebesgue measure, i.e., all the variables in the data are continuous.
The general case can be handled similarly at the expense of more cumbersome notation.

Denote by $\mathcal{Z} \subset \R^{2d}$ the support of $f_1\times f_2$ and let $\mathcal{K}$ be the set of uniformly bounded, additively separable functions on $\mathcal{Z}$, i.e.,
\[
\mathcal{K} \bydef \left\{ k:\mathcal{Z} \to \R \text{ s.t. } \|k\|_\infty \le K \text{ and } k(z_1,z_2)=\tilde k_1(z_1)+ \tilde k_2(z_2) \text{ for some } \tilde k_1, \tilde k_2 \right\},
\]
where $K>0$ is a finite constant and $\|k\|_\infty = \sup_{z\in\mathcal{Z}} |k(z)|$ is the uniform norm.
Notice that the functions  $\tilde k_1$ and $\tilde k_2$ are unique only up to addition and subtraction of a constant, but this does not play any role in our arguments.

We make the following assumptions.

\begin{assumption}[additive nonignorability + exponential link function]\label{as:AN-alt}
    The function $\tilde k_0(z_1,z_2) \bydef \log (f(z_1,z_2)/f^w(z_1,z_2))$ belongs to the set $\mathcal{K}$.
\end{assumption}

\begin{assumption}[identification of $\theta$]\label{as:identification}
\begin{subassumption}
    \item The parameter space $\Theta \subset \R^{d_{\theta}}$ is compact.
    \item $\E_f [\varphi(Z_1,Z_2,\theta)]=0$ implies $\theta=\theta_0$.
    \item $\varphi(z_1,z_2,\theta)$ is twice continuously differentiable in $\theta$ on $\Theta$ for all $z_1,z_2$.
    \item $\E \left[\frac{\partial}{\partial \theta'}\varphi(Z_1,Z_2,\theta_0)\right]$ is nonsingular.
    \item $\varphi(z,\theta)$ and $\frac{\partial \varphi}{\partial\theta'}(z,\theta)$ are bounded on $\mathcal{Z} \times \Theta$.
\end{subassumption}
\end{assumption}

\begin{assumption}[densities]\label{as:densities-alt}
\begin{subassumption}
    \item The densities $f^w(z_1,z_2;\gamma_w),f_1(z_1;\gamma_1),f_2(z_2;\gamma_2)$ belong to a parametric family and are continuously differentiable with respect to their parameters $\gamma_1,\gamma_2, \gamma_w$ that are identifiable.
    \item The gradients of densities are integrable locally uniformly in the parameters in the following sense: there exists a constant $C$ such that
    \begin{align*}
        \max\left\{\iint\left\|\frac{\partial f^w}{\partial \gamma_w}(z_1,z_2;\bar\gamma_w)\right\|dz_1 dz_2, \,\, \int\left\|\frac{\partial f_1}{\partial \gamma_1}(z_1;\bar\gamma_1)\right\|dz_1, \,\,
        \int\left\|\frac{\partial f_2}{\partial \gamma_2}(z_2;\bar\gamma_2)\right\|dz_2 \right\} \le C
    \end{align*}
    for all $\bar\gamma_w,\bar\gamma_1,\bar\gamma_2$ in a neighborhood of the population values of $\gamma_w,\gamma_1,\gamma_2$.
\end{subassumption}
\end{assumption}

\Cref{as:AN-alt} is equivalent to imposing both \Cref{as:AN,as:AN-exp}.
\Cref{as:identification} is a standard assumption ensuring identification of $\theta$.
\Cref{as:densities-alt}(i) imposes a parametric model on the observed data, thereby restricting the applicability of our results to models where the densities are known to the researcher up to a finite number of parameters.
A sufficient condition for \Cref{as:densities-alt}(ii) is that the information matrices associated with the parametric densities $f^w,f_1,f_2$ are bounded in some neighborhood of the population values of $\gamma_w,\gamma_1,\gamma_2$.

We estimate $\gamma_1,\gamma_2, \gamma_w$ by maximum likelihood and denote the corresponding density estimators $\hat f^w(z_1,z_2)=f^w(z_1,z_2;\hat \gamma_{w})$, $\hat f_1(z_1)=f_1(z_1;\hat \gamma_{1})$, $\hat f_2(z_2)=f_2(z_2;\hat \gamma_{2})$.
We omit the arguments $\gamma_1,\gamma_2,\gamma_w$ when the densities are evaluated at their population parameter values.

\begin{assumption}[maximum likelihood estimators]\label{as:mle}
$\hat\gamma_1, \hat\gamma_2,\hat\gamma_w$ are $\sqrt{n}$-consistent with the expectation bounds
    \begin{align*}
        \E \|\hat\gamma_1-\gamma_1\| = O(1/\sqrt{n}), \quad \E \|\hat\gamma_2-\gamma_2\| = O(1/\sqrt{n}), \quad \E \|\hat\gamma_w-\gamma_w\| = O(1/\sqrt{n}).
    \end{align*}
\end{assumption}

A key idea allowing us to establish the $\sqrt{n}$-convergence of $\hat f$ and $\hat\theta$ is to switch from the primal characterization of $f$ via a \emph{constrained} program \eqref{eq:KL-representation} to the dual characterization \eqref{eq:key-identity-exp} with $k_1,k_2$ solving an \emph{unconstrained} program.
This duality result is not new (see, e.g., Section 4 in \cite{bhattacharya2006iterative} and Theorem 2.3 in \cite{bhattacharya1995general}), but we provide it with a proof as Lemma 1 in the Appendix to keep the presentation self-contained.

\begin{thm}\label{thm:convergence-rate}
    Suppose that \cref{as:AN-alt,as:identification,as:abs-cont,as:densities-alt,as:mle} hold.
    Then $\|\hat f-f\|_1 = O_p \left (1/\sqrt n \right)$ and $\hat\theta-\theta_0 = O_p\left (1/\sqrt n  \right )$, where $||f||_1 =\int |f(z)| dz$.
\end{thm}
\begin{proof}
    See \Cref{app:thm2}.
\end{proof}

%% file: sections/inference.tex
\section{Inference}\label{sec:inference}

The asymptotic distribution of the moment estimator of $\theta$ depends on the asymptotic distribution of the estimator (\ref{eq:KL-representation-1}) of the pseudo-parameter $f$ defined in (\ref{eq:KL-representation}).
That estimator is computed by a finite number of iterations as in \cref{thm:raking}.
The derivation of the asymptotic distribution is analogous to that of an MLE computed by a recursive numerical optimization algorithm.
% The first-order condition is solved for the MLE, and the asymptotic distribution of this solution is the asymptotic distribution of the MLE.
% The asymptotic distribution of the estimator (\ref{eq:KL-representation-1}) is therefore equal to the $T$-fold recursion in \cref{thm:raking}. 

Under \Cref{as:densities-alt}, the marginal densities of $z_1,z_2$ and the joint density of the observed $z_1,z_2$ are in a parametric family, with the parameters that can be estimated by MLE. The asymptotic distribution of the estimator of $f$ can therefore be derived using the delta method. Define $\gamma=(\gamma_1',\gamma_2',\gamma_w')'$ and let $D_\gamma$ be the derivative operator with respect to $\gamma$. Define $m(a,b,c) \bydef ab/c$ and let $m_1=b/c$, $m_2=a/c$, and $m_3=-ab/c^2$ be its derivatives with respect to $a,$ $b$, and $c$, respectively. Let $V(\hat\gamma)$ be the asymptotic variance of an estimator $\hat\gamma$ of $\gamma$, such as the inverse information matrix when $\hat\gamma$ is the maximum likelihood estimator.

\begin{thm}\label{thm:asy-var}
    The asymptotic variance of $\hat\theta$ is $V(\hat\theta) = A_0 \cdot V(\hat\gamma) \cdot A_0'$, where 
    \begin{align*}
        A_0 = \left( D_\theta \iint \varphi(z_1,z_2;\theta_0) f(z_1,z_2) \, dz_1 dz_2 \right)^{-1} \iint \varphi(z_1,z_2;\theta_0) D_\gamma f(z_1,z_2) \,dz_1 dz_2.
    \end{align*}
    The matrix $A_0$ can be estimated by
    \begin{align*}
        \hat A_0 = \left( D_\theta \iint \varphi(z_1,z_2;\hat \theta) \hat f(z_1,z_2) \, dz_1 dz_2 \right)^{-1} \iint \varphi(z_1,z_2;\hat \theta) D_\gamma \hat f(z_1,z_2) \,dz_1 dz_2.
    \end{align*}
    Here $\hat f(z_1,z_2)=\hat f^{(T)}(z_1,z_2)$ is the $T$-th iterate of
    \begin{equation}\label{fhat-iter}
        \hat f^{(t)}(z_1,z_2)= \frac{f_2(z_2;\hat \gamma_2)m \left (f_1(z_1;\hat \gamma_1), \hat f^{(t-1)}(z_1,z_2), \Pi_1 \hat f^{(t-1)}(z_1,z_2)\right )}{\Pi_2 m\left ( f_1(z_1; \hat \gamma_1), \hat f^{(t-1)}(z_1,z_2), \Pi_1 \hat f^{(t-1)}(z_1,z_2) \right )}
    \end{equation}
with $\hat f^{(0)}(z_1,z_2)=f^w(z_1,z_2; \hat \gamma_w)$,
\[
D_{\gamma} \hat f(z_1,z_2)=D_{\gamma}\hat f^{(T)}(z_1,z_2)
\]    
and 
\[
D_{\gamma}\hat f^{(t)}(z_1,z_2)= \left ( \begin{array}{c}
   D_{\gamma_1}\hat f^{(t)}(z_1,z_2)   \\
      D_{\gamma_2}\hat f^{(t)}(z_1,z_2)\\
       D_{\gamma_w}\hat f^{(t)}(z_1,z_2) 
\end{array}  \right )
\]
 with iteration
 \begin{align*}
 &D_{\gamma_1}\hat f^{(t)}(z_1,z_2) = \\
 &\frac{f_2(z_2;\hat \gamma_2) \left ( m_1 (z_1,z_2) D_{\gamma_1}f_1(z_1;\hat \gamma_1)+m_2 (z_1,z_2 ) D_{\gamma_1}\hat f^{(t-1)}(z_1,z_2)+m_3(z_1,z_2)\Pi_1 D_{\gamma_1}\hat f^{(t-1)}(z_1,z_2)\right )}{\Pi_2 m\left ( f_1(z_1; \hat \gamma_1), \hat f^{(t-1)}(z_1,z_2), \Pi_1 \hat f^{(t-1)}(z_1,z_2) \right )}-\\
 &\frac{f_2(z_2;\hat \gamma_2)m \left (f_1(z_1;\hat \gamma_1), \hat f^{(t-1)}(z_1,z_2), \Pi_1 \hat f^{(t-1)}(z_1,z_2)\right )}{\left (\Pi_2 m\left ( f_1(z_1; \hat \gamma_1), \hat f^{(t-1)}(z_1,z_2), \Pi_1 \hat f^{(t-1)}(z_1,z_2) \right )\right)^2} \times \\
 &\Pi_2\left ( m_1 (z_1,z_2) D_{\gamma_1}f_1(z_1;\hat \gamma_1)+m_2 (z_1,z_2 ) D_{\gamma_1}\hat f^{(t-1)}(z_1,z_2)+m_3(z_1,z_2)\Pi_1 D_{\gamma_1}\hat f^{(t-1)}(z_1,z_2)\right),
 \end{align*}
 \begin{align*}
     &D_{\gamma_2}\hat f^{(t)}(z_1,z_2) = \\
 &\frac{m \left (f_1(z_1;\hat \gamma_1), \hat f^{(t-1)}(z_1,z_2), \Pi_1 \hat f^{(t-1)}(z_1,z_2)\right ) D_{\gamma_2}f_2(z_2;\hat \gamma_2)}{\Pi_2 m\left ( f_1(z_1; \hat \gamma_1), \hat f^{(t-1)}(z_1,z_2), \Pi_1 \hat f^{(t-1)}(z_1,z_2) \right )}+\\
&\frac{f_2(z_2;\hat \gamma_2)\left( m_2(z_1,z_2)D_{\gamma_2}\hat f^{(t-1)}(z_1,z_2)+m_3(z_1,z_2)\Pi_1 D_{\gamma_2}\hat f^{(t-1)}(z_1,z_2)\right)}{\Pi_2 m\left ( f_1(z_1; \hat \gamma_1), \hat f^{(t-1)}(z_1,z_2), \Pi_1 \hat f^{(t-1)}(z_1,z_2) \right )}-\\
 &\frac{f_2(z_2;\hat \gamma_2)m \left (f_1(z_1;\hat \gamma_1), \hat f^{(t-1)}(z_1,z_2), \Pi_1 \hat f^{(t-1)}(z_1,z_2)\right )}{\left (\Pi_2 m\left ( f_1(z_1; \hat \gamma_1), \hat f^{(t-1)}(z_1,z_2), \Pi_1 \hat f^{(t-1)}(z_1,z_2) \right )\right)^2} \times \\
&\Pi_2\left (m_2 (z_1,z_2 ) D_{\gamma_2}\hat f^{(t-1)}(z_1,z_2)+m_3(z_1,z_2)\Pi_1 D_{\gamma_2}\hat f^{(t-1)}(z_1,z_2)\right),
 \end{align*}
 \begin{align*}
 &D_{\gamma_w}\hat f^{(t)}(z_1,z_2) = \\
 &\frac{f_2(z_2;\hat \gamma_2) \left ( m_2 (z_1,z_2 ) D_{\gamma_w}\hat f^{(t-1)}(z_1,z_2)+m_3(z_1,z_2)\Pi_1 D_{\gamma_w}\hat f^{(t-1)}(z_1,z_2)\right )}{\Pi_2 m\left ( f_1(z_1; \hat \gamma_1), \hat f^{(t-1)}(z_1,z_2), \Pi_1 \hat f^{(t-1)}(z_1,z_2) \right )}-\\
 &\frac{f_2(z_2;\hat \gamma_2)m \left (f_1(z_1;\hat \gamma_1), \hat f^{(t-1)}(z_1,z_2), \Pi_1 \hat f^{(t-1)}(z_1,z_2)\right )}{\left (\Pi_2 m\left ( f_1(z_1; \hat \gamma_1), \hat f^{(t-1)}(z_1,z_2), \Pi_1 \hat f^{(t-1)}(z_1,z_2) \right )\right)^2} \times \\
 &\Pi_2\left ( m_2 (z_1,z_2 ) D_{\gamma_w}\hat f^{(t-1)}(z_1,z_2)+m_3(z_1,z_2)\Pi_1 D_{\gamma_w}\hat f^{(t-1)}(z_1,z_2)\right).
 \end{align*}
In the above equations,
\[
m_k(z_1,z_2)=m_k (f_1(z_1;\hat \gamma_1), \hat f^{(t-1)}(z_1,z_2), \Pi_1 \hat f^{(t-1)}(z_1,z_2)), \ \ \  k=1,2,3,
\]
and the initial conditions are
\[
D_{\gamma}\hat f^{(0)}(z_1,z_2)= \left ( \begin{array}{c}
   D_{\gamma_1}\hat f^{(0)}(z_1,z_2)   \\
      D_{\gamma_2}\hat f^{(0)}(z_1,z_2)\\
       D_{\gamma_w}\hat f^{(0)}(z_1,z_2) 
\end{array}  \right )=  
\left ( \begin{array}{c} 
0\\
0\\
D_{\gamma_w}f^w(z_1,z_2; \hat \gamma_w)\\
\end{array}  \right ).
\]

\end{thm}

\begin{proof}
   Differentiation of (\ref{fhat-iter}) with respect to $\gamma$ yields the iteration. 
\end{proof}

%% file: sections/implementation.tex
\section{Numerical implementation}\label{sec:numerical}

\subsection{The general case}

Both raking iterations and moment conditions require computing integrals.
To do it quickly and reliably, we suggest using Monte Carlo integration on a fixed rectangular grid in terms of the variables $z_1,z_2$ that we generate randomly from distributions that dominate $Z_1$ and $Z_2$.

To describe our procedure, suppose both $f_1$ and $f_2$ are absolutely continuous with respect to a dominating measure $\mu$. For instance, when the data are one discrete and one continuous variable, $\mu$ is the product of a counting measure and the Lebesgue measure. Notice that integrals in raking iterations and moment conditions are calculated with respect to $\mu$.

Let $\phi_1$ and $\phi_2$ be densities on $\R^d$ such that $f_1$ and $f_2$ are absolutely continuous w.r.t. $\phi_1$ and $\phi_2$, respectively.
Typically, $\phi_1=\phi_2$.
Let $z_{11},\dots, z_{1S_1}$ and $z_{21},\dots,z_{2S_2}$ denote independent draws from distributions $\phi_1$ and $\phi_2$, respectively.
Fix the rectangular grid $\{z_{11},\dots,z_{1S_1}\} \times \{z_{21},\dots,z_{2S_2}\}$ from now on. 
We only need to calculate the integrals involved in raking and possibly moment conditions on this grid.

For illustration, suppose $Z=(Y,X)'$, where $X$ is discrete with support $\mathcal{X} = \{x_1,\dots,x_R\}$ and $Y$ is continuous.
Let $\phi_X$ be a probability mass function on $\mathcal{X}$, i.e. $\phi_X(x) >0$ for $x\in\mathcal{X}$, $\phi_X(x)=0$ for $x\notin \mathcal{X}$, and $\phi_X(x_1)+\dots+\phi_X(x_R)=1$.
Let $\phi_Y$ be a probability density on $\R^1$, e.g., a Gaussian density.
Then, for period 1, we can independently sample $S_1$ draws from $\phi_X$ and $S_1$ draws from $\phi_Y$, obtaining a combined sample $\mathcal{Z}_1= \{(x_{11},y_{11}),\dots,(x_{1S_1},y_{1S_1})\}$.
Similarly, for period 2, we obtain a sample $\mathcal{Z}_2= \{(x_{21},y_{21}),\dots,(x_{2S_2},y_{2S_2})\}$ using the same $\phi_X$ and $\phi_Y$. Then $\phi_1=\phi_2$ and the grid is $\mathcal{Z}_1 \times \mathcal{Z}_2$.
The density $\phi=\phi_1=\phi_2$ with respect to the product of the counting measure on $\mathcal{X}$ and the Lebesgue measure on $\R^1$ is
\[
\phi(z)=\phi_X(x)\phi_Y(y) \text{ for } z=(x,y) \in \R^2.
\]

We now describe our integral approximations. The first step of our procedure, raking, requires approximations of integrals with respect to one period variable. For example, an odd iteration of raking is
\begin{align*}
    f^{(t+1)}(z_{1i},z_{2j}) = \frac{f_1(z_{1i}) f^{(t)}(z_{1i},z_{2j})}{\int f(z_{1i},z_2) \, d\mu(z_2) }.
\end{align*}
Here, we approximate the integral in the denominator as 
\[
\int f(z_{1i},z_2) \, d\mu(z_2) \approx \frac{1}{S_2} \sum_{j=1}^{S_2} \frac{f(z_{1i},z_{2j})}{\phi_2(z_{2j})}.
\]
Similarly, for an even iteration, we approximate the denominator as 
\[
\int f(z_1,z_{2j}) \, d\mu(z_1) \approx \frac{1}{S_1} \sum_{i=1}^{S_1} \frac{f(z_{1i},z_{2j})}{\phi_1(z_{1i})}.
\]

The second step of our procedure, calculating the moment conditions, requires the approximation of integrals with respect to the raked distribution.
If the model is implemented reliably in software packages (say, two-way fixed effects regression), we can draw a very large sample from the raking estimator $\hat f$ and treat it as the input to the existing routine.
This will then be numerically close to solving the moment conditions with plugged-in $\hat f$.
On the other hand, if calculating the moment conditions is necessary, we approximate $\E_f \varphi(Z_1,Z_2)$ as
\begin{align*}
    \iint \varphi(z_1,z_2) f(z_1,z_2) \, d\mu(z_1) d\mu(z_2) &\approx \frac{1}{S_1 S_2} \sum_{i=1}^{S_1} \sum_{j=1}^{S_2} \frac{\varphi(z_{1i},z_{2j}) f(z_{1i},z_{2j})}{\phi_1(z_{1i}) \phi_2(z_{2j})}.
\end{align*}

\subsection{The Gaussian case}\label{sec:gaussian-raking}

If input distributions $f_1$, $f_2$, and $f^w$ are all modeled as Gaussian, the raking algorithm becomes particularly simple, because the raking iterations can be written explicitly in terms of Gaussian means and covariance matrices.
Specifically, suppose $g$ is the $2d$-dimensional distribution $N(\nu,C)$ and $f_1$ is the $d$-dimensional distribution $N(\mu_1,\Sigma_1)$, with the conformable split
\begin{align*}
    \nu = \begin{bmatrix}\nu_1 \\ \nu_2\end{bmatrix},
    \quad 
    C = \begin{bmatrix}C_{11} & C_{12} \\ C_{21} & C_{22}\end{bmatrix}.
\end{align*}
Then it is easy to show that the raking iteration $\phi_{1,f_1}(g)$ (i.e., the KL projection of $g$ on the set of distributions with first $d$-dimensional marginal $f_1$) is the Gaussian distribution with mean
\begin{align*}
\nu^* =
\begin{bmatrix}
\mu_1 \\
\nu_2 + C_{21}C_{11}^{-1}(\mu_1 - \nu_1)
\end{bmatrix}
\end{align*}
and the covariance matrix
\begin{align*}
    C^* =
\begin{bmatrix}
\Sigma_1 & \Sigma_1 C_{11}^{-1} C_{12} \\
C_{21} C_{11}^{-1} \Sigma_1 &
C_{22} - C_{21}C_{11}^{-1}C_{12}
+ C_{21}C_{11}^{-1}\Sigma_1C_{11}^{-1}C_{12}
\end{bmatrix}.
\end{align*}
Similarly, suppose that $f_2$ is the $d$-dimensional distribution $N(\mu_2,\Sigma_2)$.
Then the raking iteration $\phi_{2,f_2}(g)$ (i.e., the KL projection of $g$ on the set of distributions with second $d$-dimensional marginal $f_2$) is the Gaussian distribution with mean
\begin{align*}
\nu^* =
\begin{bmatrix}
\nu_1 + C_{12}C_{22}^{-1}(\mu_2-\nu_2)\\
\mu_2
\end{bmatrix}    
\end{align*}
and the covariance matrix
\begin{align*}
    C^* =
\begin{bmatrix}
C_{11} - C_{12}C_{22}^{-1}C_{21}
+ C_{12}C_{22}^{-1}\Sigma_2 C_{22}^{-1}C_{21}
&
C_{12}C_{22}^{-1}\Sigma_2
\\
\Sigma_2 C_{22}^{-1}C_{21}
&
\Sigma_2
\end{bmatrix}.
\end{align*}
Using these formulas recursively leads to the Gaussian raking algorithm that is easy to code and tends to converge very quickly.

%% file: sections/MC.tex
\section{Monte Carlo simulation}\label{sec:mc}

In this section, we illustrate the performance of our estimation procedure in a set of Monte Carlo simulations.
We employ two data-generating processes (DGPs) with discrete and continuous variables.

The discrete DGP is a Markov chain with states $\{0,\dots,d-1\}$, where $Z_1$ is distributed uniformly over the $d$ states ($d=5$ or $d=10$) and $Z_2$ is generated from $Z_1$ using a transition matrix with all the transition probabilities within the interval $[0.05, 0.95]$. From the marginal distribution of $Z_1$ and the conditional distribution of $Z_2|Z_1=z_1$, we calculate the marginal distribution of $Z_2$ that is used to draw the refreshment sample.
The parameter of interest is
\begin{align*}
    \theta = \E [Z_1 Z_2],
\end{align*}
and the selection probability $P(W=1|Z_1=z_1,Z_2=z_2)= \exp(k_1(z_1)+k_2(z_2))$ with
\begin{align*}
    k_1(z_1) = -0.1(1+0.5 z_1), \quad k_2(z_2) = -0.2(1+0.5 z_2).
\end{align*}
Notice that these functions satisfy $k_1(z_1)+k_2(z_2)\le 0$ for all $z_1,z_2$ in the supports of $Z_1,Z_2$.
The overall attrition rate is 40\% for $d=5$ and 53\% for $d=10$.
We assume that the support of the joint distribution of $Z_1, Z_2$ is known.
We use frequency estimators of the marginal probability mass functions of $Z_1$ and $Z_2$ and the selective joint probability mass function of $Z_1, Z_2$ as input to the raking procedure.

The continuous DGP models $(Z_1,Z_2)$ as a two-dimensional Gaussian vector with zero means, unit variances, and $\theta = Cov(Z_1,Z_2)=0.4.$
As before, the attrition process is additively nonignorable with
\begin{align*}
    k_1(z_1) = -0.1 |z_1|, \quad k_2(z_2) = -0.3 |z_2|,
\end{align*}
satisfying $k_1(z_1)+k_2(z_2)\le 0$ for all $z_1,z_2\in\R$.
The overall attrition rate is about 30\%. The marginal distributions of $Z_1$ and $Z_2$ and the balanced panel distribution of $Z_1,Z_2$ are specified as normal with means and variances estimated by MLE.
Notice that this model is misspecified: although $(Z_1,Z_2)$ is Gaussian, the balanced panel distribution $Z_1,Z_2|W=1$ is not due to $k_1(z_1),k_2(z_2)$ depending on the absolute values of $z_1,z_2$.

We report bias, standard deviation (sd), and root-MSE (rmse) of three estimators.
The infeasible estimator $\hat\theta_{\text{infeas}}$ uses the true values $k_1,k_2$ to weight the balanced panel sample. The naive estimator $\hat\theta_{\text{naive}}$ disregards attrition and treats the balanced panel as the full panel, and therefore is inconsistent. Neither of these two estimators uses the refreshment sample.
Finally, $\hat\theta$ is our procedure that uses raking to estimate the balanced panel weights.

\Cref{tab-mc-discr,tab-mc-cont} contain the simulation results.
We denote by $N$ the sum of the sample sizes of the incomplete panel and the refreshment sample, with the latter set at $0.4 N$.
The simulations for the discrete DGP use the classical raking procedure on a $d
\times d$ contingency table.
The simulations for the continuous DGP use the Gaussian version of the raking procedure; see \Cref{sec:gaussian-raking}.
In all cases, we stop the raking procedure after 50 iterations, which suffices for convergence.
Our estimator performs well in terms of both bias and RMSE.
It typically has slightly higher bias than the infeasible estimator, but a lower RMSE, partially driven by its use of the refreshment sample.

\begin{table}[h!]
\centering
\small

\begin{tabular}{c c c c c c c c c c c}
\toprule
& & \multicolumn{3}{c}{$N=1000$} & \multicolumn{3}{c}{$N=5000$} & \multicolumn{3}{c}{$N=10000$}\\	
\cmidrule(lr){3-5} \cmidrule(lr){6-8} \cmidrule(lr){9-11}	
&	& $\hat\theta_{\text{infeas}}$ &	$\hat\theta_{\text{naive}}$ &	$\hat\theta$ 
    & $\hat\theta_{\text{infeas}}$ &	$\hat\theta_{\text{naive}}$ &	$\hat\theta$
    & $\hat\theta_{\text{infeas}}$ &	$\hat\theta_{\text{naive}}$ &	$\hat\theta$ \\
\hline
\multirow{3}{*}{$d=5$}	
    & bias	& 0.002 & -0.46 & -0.007 & 0.009 & -0.45 & -0.0008 & -0.0006 & -0.46 & -0.003 \\ 
	& sd	& 0.283 & 0.18 &  0.206 &  0.125 &  0.09 & 0.0920 &   0.0829 & 0.05 & 0.063 \\
	& rmse  & 0.283 & 0.50 &  0.207 &  0.125 &  0.46 & 0.0920 &   0.0829 & 0.46 & 0.064 \\
 \midrule
\multirow{3}{*}{$d=10$}	
    &	bias & -0.001 & -3.92 & -0.085 & 0.074 & -3.88 & -0.0095 &  -0.0007 & -3.91 & -0.020 \\
	&	sd   & 1.653 & 0.82  &   0.996 & 0.751 & 0.37 & 0.4169 & 0.5044 &  0.25 &  0.303\\
	&	rmse & 1.653 & 4.01  &   1.000 & 0.755 & 3.90 & 0.4171 & 0.5044 & 3.92 & 0.304\\
 \bottomrule
\end{tabular}
\caption{Simulation results for the discrete DGP. Number of replications $S=1000$.}
\label{tab-mc-discr}
\end{table}

\begin{table}[h!]
\centering
\small

\begin{tabular}{c c c c c c c}
\toprule
& \multicolumn{3}{c}{$N=1000$} & \multicolumn{3}{c}{$N=5000$} \\
\cmidrule(lr){2-4} \cmidrule(lr){5-7} 
    & $\hat\theta_{\text{infeas}}$ &	$\hat\theta_{\text{naive}}$ &	$\hat\theta$ 
    & $\hat\theta_{\text{infeas}}$ &	$\hat\theta_{\text{naive}}$ &	$\hat\theta$ \\
\midrule
    bias & -0.012 & -0.109 & -0.018 & -0.009 & -0.106 & 0.002   \\ 
	sd	& 0.062 & 0.041 & 0.054     &  0.032 & 0.022 & 0.029 \\ 
	rmse  & 0.063 & 0.116 & 0.057   &  0.034 & 0.109 & 0.029  \\ 
 \bottomrule
\end{tabular}
\caption{Simulation results for the continuous DGP. Number of replications $S=1000$.}
\label{tab-mc-cont}
\end{table}

The standard deviations for the continuous DGP ($0.054$ and $0.029$) are close to the theoretical values ($0.052$ and $0.023$), which are derived in Appendix \ref{app:MC_var_theory}.

%% file: sections/empirical.tex
\section{Empirical illustration}\label{sec:empirical}

\begin{table}[t]
\centering
\small
\begin{tabular}{@{}llllll@{}}
\toprule
Population              & Frequency & \multicolumn{2}{c}{$5\%$ quantile}  & \multicolumn{2}{c}{$95\%$ quantile} \\
                        &           & Numeracy 1 & Numeracy 2 & Numeracy 1 & Numeracy 2 \\ \midrule
Balanced Panel (BP)     & 3303      & 38               & 38               & 65               & 65               \\
Incomplete Panel (IP)   & 904       &                  &                  &                  &                  \\
Complete Panel (CP)     & 4207      & 38               &                  & 65               &                  \\
Refreshment Sample (RS) & 3180      &                  & 34               &                  & 66               \\ \bottomrule
\end{tabular}
\caption{Sample sizes for the Balanced Panel, the Incomplete Panel, and the Refreshment Sample, as well as $5\%$ and $95\%$ quantiles of the dependent variable Numeracy in both waves.}
\label{tab:missing_data_pattern}
\end{table}

The empirical infrastructure that spawned the analysis in this paper is the Understanding America Study (UAS).\footnote{See \href{https://uasdata.usc.edu/index.php}{https://uasdata.usc.edu/index.php}}
The UAS is a probability-based Internet panel of about 15,000 respondents.
Respondents without prior internet access receive a tablet and a broadband internet connection. Participants complete online surveys once or twice monthly.
Twenty-four core modules fielded biennially (about 400 total minutes) cover physical and mental health, economics, cognition, decision-making, and social determinants. Short monthly health updates capture acute and dynamic events. The system also supports high-frequency and event-triggered assessments (up to 6x/day) for ecological or crisis-response studies.
In addition, the UAS collects genetic information and some digital biomarkers.
Like any longitudinal study, the UAS experiences attrition (7-8 percent per year); see \citet{kapteyn2024understanding}.

To counter the effect of attrition, and because the panel is still growing, refreshment samples are drawn regularly (often several times a year). The current paper is a first step in a program to optimally use the refreshment samples for statistical inference.

The illustration in this paper is taken from the first two biennial waves of the UAS (May 20, 2015 - June 1, 2017 and June 1, 2017 - June 18, 2019). The frequencies listed in Table \ref{tab:missing_data_pattern} show that 904 of the initial 4207 respondents dropped out in the second wave, which amounts to a $21.5\%$ attrition rate. For expository reasons, we consider a very simple model of cognition. Most cognitive dimensions tend to decrease with age. In addition, the literature has identified a large number of risk factors that increase the chance of dementia and generally may hasten cognitive decline. Here, we consider diabetes and depression as risk factors for cognitive decline. Education is generally found to be protective (\citet{livingston2024dementia}). 

We use numeracy as a simple measure of cognition. Table \ref{tab:missing_data_pattern} also reports the $5\%$ and $95\%$ quantiles of numeracy in different subsamples of the data. This shows that numeracy has approximately equal sample support in the balanced panel, the complete panel, and the refreshment sample, which is of importance for the raking procedure. Appendix \ref{app:variables} contains a description of the variables used in the empirical analysis. The numeracy scores are the result of answers to eight questions and are calculated using a two-parameter logistic IRT (Item Response Theory) model. Scores are normalized to a mean of 50 and a standard deviation of 10. Higher scores indicate better performance on numeracy problems.

The first five columns of Table \ref{tab:primafacie} display the sample averages of the covariates over the different subsamples.\footnote{The panel and refreshment samples were obtained using stratified sampling. The analyses in this section account for this by using sampling weights.} In these columns, we can see that the Balanced Panel contains $75.6\%$ white respondents, while the refreshment sample contains $14.4$ percentage points fewer white respondents. The refreshment sample contains more depressed respondents, fewer respondents with a college degree, and fewer respondents with diabetes, compared to the balanced panel. The dropouts in the incomplete panel are less white, more depressed, and have a lower prevalence of diabetes. Column (6) shows that dropouts have lower numeracy (as expected), and that numeracy is much lower in the balanced panel than in the refreshment sample. This holds even after correcting for covariates in column (7).
Although these differences are not conclusive, they indicate that attrition may lead to invalid inference when attrition is ignored by using only the balanced panel.

\begin{table}[]
\centering
\small
{
\def\sym#1{\ifmmode^{#1}\else\(^{#1}\)\fi}
\begin{tabular}{l*{7}{c}}
\toprule
            &\multicolumn{1}{c}{(1)}&\multicolumn{1}{c}{(2)}&\multicolumn{1}{c}{(3)}&\multicolumn{1}{c}{(4)}&\multicolumn{1}{c}{(5)}&\multicolumn{1}{c}{(6)}&\multicolumn{1}{c}{(7)}\\
            &\multicolumn{1}{c}{White}&\multicolumn{1}{c}{Depressed}&\multicolumn{1}{c}{College}&\multicolumn{1}{c}{Diabetes}&\multicolumn{1}{c}{Age}&\multicolumn{1}{c}{Numeracy}&\multicolumn{1}{c}{Numeracy}\\
\hline
IP       &      -0.033\sym{**}         &       0.054\sym{***}         &      -0.020         &      -0.020\sym{*}  &      -3.184\sym{***}         &      -0.634\sym{*}  &   -0.313   \\
            &      (0.01)         &      (0.01)         &      (0.01)         &      (0.01)         &      (0.78)         &      (0.29)         &      (0.30)         \\
RS       &      -0.144\sym{***}&       0.453\sym{***}&      -0.421\sym{***}&      -0.015\sym{*}&       -0.182         &       0.907\sym{***}&       0.729\sym{**}  \\
            &      (0.01)         &      (0.01)         &      (0.01)         &      (0.01)         &      (0.43)         &      (0.19)         &      (0.27)         \\
Age    &                     &                     &                     &                     &                     &                     &      -0.023\sym{***}\\
            &                     &                     &                     &                     &                     &                     &      (0.01)         \\
Diabetes    &                     &                     &                     &                     &                     &                     &      -1.513\sym{***}\\
            &                     &                     &                     &                     &                     &                     &      (0.28)         \\
College     &                     &                     &                     &                     &                     &                     &       5.161\sym{***}\\
            &                     &                     &                     &                     &                     &                     &      (0.21)         \\
Depressed   &                     &                     &                     &                     &                     &                     &      -2.233\sym{***}\\
            &                     &                     &                     &                     &                     &                     &      (0.21)         \\
White       &                     &                     &                     &                     &                     &                     &       4.168\sym{***}\\
            &                     &                     &                     &                     &                     &                     &      (0.22)         \\
Constant  &       0.756\sym{***}&       0.247\sym{***}&       0.728\sym{***}&       0.131\sym{***}&      53.216\sym{***}&      50.367\sym{***}&      45.327\sym{***}\\
            &      (0.01)         &      (0.01)         &      (0.01)         &      (0.00)         &      (0.27)         &      (0.11)         &      (0.36)         \\
\bottomrule
\multicolumn{8}{l}{\footnotesize Standard errors in parentheses}\\
\multicolumn{8}{l}{\footnotesize \sym{*} \(p<0.05\), \sym{**} \(p<0.01\), \sym{***} \(p<0.001\)}\\
\end{tabular}
}
\caption{Regressions of the observed variables on population indicator variables. The balanced panel (BP) is the reference group.}
\label{tab:primafacie}
\end{table}

We estimate a random-effects linear regression with random effect $u_i$ and idiosyncratic error $e_{it}$, using the standard random-effects error structure with variances $\sigma^2_u$ and $\sigma^2_e$. We allow the coefficient vector $\beta$ to be different in period one and period two. 

Table \ref{tab:primafacie} shows that the distribution of the observed variables in the refreshment sample is substantially different from the distribution of these variables in the balanced panel. The additively nonignorable attrition model finds the population distribution that is consistent with the first-period marginal, obtainable from the balanced panel and the incomplete panel, and the second-period marginal, obtainable from the refreshment sample, using the raking weights.  
The attrition correction can be expected to affect the conditional mean function of numeracy given the covariates differentially in the two periods.

Table \ref{tab:estimates} shows estimates and standard errors of this model when attrition is ignored (Missing Completely At Random, MCAR), under MAR (Missing At Random) as well as when corrected for possibly nonignorable attrition using the AN model. The estimates are obtained using a weighted GMM procedure using the raking weights. The standard errors of MAR are bootstrapped. The AN standard errors are obtained using Theorem \ref{thm:asy-var}.

\begin{table}[]
\centering
\small{
\def\sym#1{\ifmmode^{#1}\else\(^{#1}\)\fi}
\begin{tabular}{@{}lcccccc@{}}
\toprule
& \multicolumn{2}{c}{(1)} & \multicolumn{2}{c}{(2)} & \multicolumn{2}{c}{(3)} \\
& \multicolumn{2}{c}{MCAR (BP)} & \multicolumn{2}{c}{MAR (BP and IP)} & \multicolumn{2}{c}{AN (BP,IP and RS)} \\ \cmidrule(lr){2-3} \cmidrule(lr){4-5} \cmidrule(l){6-7}
& Numeracy1 & Numeracy2 & Numeracy1 & Numeracy2 & Numeracy1 & Numeracy2 \\ \midrule
Age & -0.00856 & -0.00835 & -0.01035 & -0.00845 & -0.0143 & -0.00864 \\
& (0.00706) & (0.00704) & (0.00638) & (0.00687) & (0.0147) & (0.0198) \\
Diabetes & -1.256\sym{**} & -1.723\sym{***} & -1.324\sym{***} & -1.714\sym{***} & -1.363\sym{***} & -2.081\sym{*} \\
& (0.446) & (0.451) & (0.387) & (0.416) & (0.334) & (0.962) \\
College & 4.988\sym{***} & 4.718\sym{***} & 5.079\sym{***} & 4.779\sym{***} & 5.234\sym{***} & 4.784\sym{***} \\
& (0.314) & (0.320) & (0.281) & (0.303) & (0.299) & (0.828) \\
Depressed & -2.379\sym{***} & -2.237\sym{***} & -2.249\sym{***} & -2.079\sym{***} & -0.983 & 0.128 \\
& (0.339) & (0.346) & (0.300) & (0.331) & (0.630) & (0.556) \\
White & 4.124\sym{***} & 4.334\sym{***} & 3.878\sym{***} & 4.096\sym{***} & 3.907\sym{***} & 3.396\sym{***} \\
& (0.340) & (0.347) & (0.304) & (0.323) & (0.847) & (0.574) \\
Constant & 44.87\sym{***} & 44.83\sym{***} & 44.99\sym{***} & 44.90\sym{***} & 45.71\sym{***} & 46.41\sym{***} \\
& (0.466) & (0.489) & (0.408) & (0.478) & (0.934) & (1.322) \\ \midrule
$\sigma^{2}_{e}$ & \multicolumn{2}{c}{28.00\sym{***}} & \multicolumn{2}{c}{28.08\sym{***}} & \multicolumn{2}{c}{27.99\sym{***}} \\
& \multicolumn{2}{c}{(0.907)} & \multicolumn{2}{c}{(0.920)} & \multicolumn{2}{c}{(2.403)} \\
$\sigma^{2}_{u}$ & \multicolumn{2}{c}{44.79\sym{***}} & \multicolumn{2}{c}{45.07\sym{***}} & \multicolumn{2}{c}{46.38\sym{***}} \\
& \multicolumn{2}{c}{(1.221)} & \multicolumn{2}{c}{(1.158)} & \multicolumn{2}{c}{(3.037)} \\
\bottomrule
\multicolumn{7}{l}{\footnotesize Standard errors in parentheses} \\
\multicolumn{7}{l}{\footnotesize \sym{*} \(p<0.05\), \sym{**} \(p<0.01\), \sym{***} \(p<0.001\)} \\
\end{tabular}}
\caption{ \small{Estimates and Standard Errors under Missing Completely At Random (MCAR), Missing At Random (MAR), and Additively Nonignorable (AN) attrition.}}
\label{tab:estimates}
\end{table}

Note that the standard errors for the AN model are higher because it estimates the population model while taking attrition into account. In the AN model, only attending college and being white are significant covariates. Moreover, the regression coefficients of depressed and white in the two periods in the MCAR and MAR models are closer to each other than in the AN model. The Wald test does not reject equality of the regression coefficients under MCAR (MAR), with a $\chi^2_6$ statistic of $3.73$ and the $p$-value $0.71$ ($3.54$ with the $p$-value $0.74$, respectively). Under AN, this statistic equals $19.45$ (with the $p$-value $0.003$), leading to rejection of stability at conventional significance levels. Thus, the random-effects specification is rejected in this example.

Empirically, the most notable result in Table \ref{tab:estimates} is that under AN, both depression and diabetes are no longer statistically significant, whereas they are highly significant under MCAR and MAR.
% Moreover, a balanced panel analysis underestimates the effect of diabetes and college attendance on numeracy and overestimates the effect of being depressed on numeracy.
As Table \ref{tab:primafacie} makes clear, respondents in the balanced panel are less depressed, suffer more from diabetes, and are more likely to have a college degree, relative to respondents in the incomplete panel and refreshment samples. Unlike the MCAR and MAR estimates, the AN estimates are aligned with the marginal distributions of all variables in both time periods.

%% file: sections/conclusion.tex
\section{Conclusion}\label{sec:conclusion}

Estimation and inference in panels with attrition and refreshment have long lacked a tractable implementation. 
We fill this gap by showing that (a continuous version of) raking can be used to estimate the target distribution.
When the input distributions are parametric, we derive a convergence rate of our two-step estimator and provide a recursive approximation to its asymptotic variance.
Promising directions for future work include establishing bootstrap validity, deriving the semiparametric efficiency bound and a corresponding variance formula, extending the framework to multi-wave panels, and characterizing the identified set under weaker assumptions on the attrition mechanism.

\newpage

\section*{Acknowledgements}

We are grateful to Tim Armstrong, Tim Christensen, Vasily Goncharenko, Lidia Kosenkova, Sergey Lototsky, Anna Mikusheva, Francesca Molinari, Elizaveta Rebrova, Azeem Shaikh, Andrei Zeleneev, and seminar participants at USC, UC Irvine, The University of Virginia, and Penn State for valuable comments.
All errors and omissions are our own.

%% file: sections/appendix.tex
\newpage  
\part*{Appendix}
\appendix

\section{Proof of \cref{thm:raking}}\label{app:thm1}

  Theorem 3.5 in \citet{ruschendorf1995convergence} states that the measures corresponding to densities $\hat f^{(t)}$ converge to the measure corresponding to $\hat f$ in total variation. The latter implies $\|\hat f^{(t)}-\hat f\|_1 \to 0$ because the total variation distance can be written as $\operatorname{TV}(\hat f^{(t)},\hat f)=\frac{1}{2}\|\hat f^{(t)}-\hat f\|_1.$

\section{Proof of \cref{thm:convergence-rate}}\label{app:thm2}

First, we give a proof of a well-known dual representation for the KL projection problem (see, e.g., Section 4 in \cite{bhattacharya2006iterative} and Theorem 2.3 in \cite{bhattacharya1995general}).

\begin{lem}[duality]\label{lem:duality}
    Denote
    \begin{align*}
        \M(k) &\bydef \iint e^{k_1(z_1)+k_2(z_2)} f^w(z_1,z_2)\,dz_1 dz_2 - \int k_1(z_1) f_1(z_1) \, dz_1 - \int k_2(z_2) f_2(z_2) \, dz_2, \\
        \M_n(k) &\bydef \iint e^{k_1(z_1)+k_2(z_2)} \hat f^w(z_1,z_2)\,dz_1 dz_2 - \int k_1(z_1) \hat f_1(z_1) \, dz_1 - \int k_2(z_2) \hat f_2(z_2) \, dz_2.
    \end{align*}
    Then the population density $f$ (estimator $\hat f$) is related to the population density $f^w$ (estimator $\hat f^w$) through
    \begin{align}
        f(z_1,z_2) &= e^{k_0(z_1,z_2)} f^w(z_1,z_2), \label{eq:f-via-k0} \\
        \hat f(z_1,z_2) &= e^{\hat k(z_1,z_2)} \hat f^w(z_1,z_2),
    \end{align}
    where
    \begin{align}
        k_0 &= \arg\min_{k \in \mathcal{K}} \M(k), \label{eq:k0-via-dual} \\
        \hat k &= \arg\min_{k \in \mathcal{K}} \M_n(k).
    \end{align}
    
\end{lem}

\begin{proof}

    We follow Chapter 5 of \citet{boyd2004convex}.
    The program \eqref{eq:KL-representation} has the Lagrangian representation
    \begin{align}
        \min_{\tilde f} \max_{\lambda_1,\lambda_2} \mathcal{L}(\tilde f,\lambda_1,\lambda_2) \label{eq:primal}
    \end{align}
    with
    \begin{align*}
        \mathcal{L}(\tilde f,\lambda_1,\lambda_2) &= \int \tilde f(z)\log \frac{\tilde f(z)}{f^w(z)} \, dz + \int \lambda_1(z_1) \left( \int \tilde f(z_1,z_2) \,dz_2 - f_1(z_1) \right) \, dz_1  \\
        &+ \int \lambda_2(z_2) \left( \int \tilde f(z_1,z_2) \,dz_1 - f_2(z_2) \right) \, dz_2 \\
        &= \int \tilde f(z)\left[\log \frac{\tilde f(z)}{f^w(z)}+\lambda_1(z_1)+\lambda_2(z_2) \right] \,dz_1 dz_2 \\
        &-\int \lambda_1(z_1) f_1(z_1)\,dz_1 - \int \lambda_2(z_2)f_2(z_2)\,dz_2.
    \end{align*}
    The functional derivative of $\mathcal{L}$ with respect to $\tilde f(z_1,z_2)$ is
    \begin{align*}
        \frac{\partial\mathcal{L}}{\partial \tilde f(z_1,z_2)} = \log \frac{\tilde f(z)}{f^w(z)} + 1+\lambda_1(z_1)+\lambda_2(z_2).
    \end{align*}
    Setting this derivative to zero yields
    \begin{align}
        f(z_1,z_2) = f^w(z_1,z_2) e^{-1-\lambda_1(z_1)-\lambda_2(z_2)}. \label{eq:FOC-primal}
    \end{align}
    Notice that this condition is sufficient for the global minimum since the objective functional $\tilde f \mapsto \operatorname{KL}(\tilde f,f^w)$ is strictly convex. Hence, it is satisfied by the population density $f$ as reflected above.
    Substituting back into the Lagrangian yields the dual Lagrangian function \citep[Section 5.1.2]{boyd2004convex}
    \begin{align*}
        \min_{\tilde f}\mathcal{L}(\tilde f,\lambda_1,\lambda_2) = &-\int f^w(z_1,z_2) e^{-1-\lambda_1(z_1)-\lambda_2(z_2)} \, dz_1 dz_2 \\
        &- \int \lambda_1(z_1) f_1(z_1)\,dz_1 - \int \lambda_2(z_2)f_2(z_2)\,dz_2.
    \end{align*}
    
    The dual problem is then
    \begin{align}
        \max_{\lambda_1,\lambda_2} \min_{\tilde f} \mathcal{L}(\tilde f,\lambda_1,\lambda_2). \label{eq:KL-dual}
    \end{align}
    Since the primal problem is convex, the solutions for the primal problem \eqref{eq:primal} and the dual problem \eqref{eq:KL-dual} coincide \citep[Section 5.5.3]{boyd2004convex}.
    Denoting $k_1(z_1)=-1-\lambda_1(z_1)$ and $k_2(z_2)=-\lambda_2(z_2)$, we see that the dual problem is equivalent to \eqref{eq:k0-via-dual}, while the first-order condition \eqref{eq:FOC-primal} yields \eqref{eq:f-via-k0}.

    An analogous derivation can be performed for the sample version of the problem.
\end{proof}

We now proceed to the proof of Theorem 2. We first show that 
\begin{align}
    \| \hat k - k_0 \|_{\infty} = O_p(1/\sqrt n). \label{eq:k-hat-consistent}
\end{align}
Since $k_0, \hat k$ minimize the population and sample criterion functions, respectively, we can apply Theorem 3.2.5 in \citet{vdv}. By the definition of $\M_n$ and $\M$,
\begin{align*}
(\M_n(k)-\M(k))-(\M_n(k_0)-\M(k_0))&=\iint \left (e^{k(z)}- e^{k_0(z)} \right )  \left ( \hat f^w(z)-f^w(z)\right) dz \\
&+\int (k_1(z_1)-k_{10}(z_1))\left( \hat f_1(z_1)-f_1(z_1)\right )dz_1 \\
&+ \int (k_2(z_2)-k_{20}(z_2))\left( \hat f_2(z_2)-f_2(z_2)\right )dz_2.
\end{align*}
Hence
\begin{align*}
\sup_{\|k-k_0\|_{\infty}< \delta}\left | (\M_n(k)-\M(k))-(\M_n(k_0)-\M(k_0))\right | &\le e^{\sup_{k \in \mathcal{K}} \|k\|_{\infty}}\delta \iint \left |\hat f^w(z)-f^w(z)\right | dz \\
&+ \delta \int \left |\hat f_1(z_1)-f_1(z_1)\right |dz_1  \\
&+ \delta \int \left |\hat f_2(z_2)-f_2(z_2)\right |dz_2.
\end{align*}
Let us show that
\[
\sqrt{n }\E\iint \left|\hat f^w (z)-f^w(z)\right|\,dz = O(1),
\]
and similarly for $\hat f_1$ and $\hat f_2$.
Indeed, by Assumption \ref{as:densities-alt} and the mean value theorem,
\begin{align}
\hat f^w(z)-f^w(z)= f^w(z; \hat \gamma_w) -f^w(z; \gamma_w)  =\frac{\partial f^w}{\partial \gamma_w'}(z; \overline \gamma_w)(\hat \gamma_w - \gamma_w).    \label{eq:fw-hat-taylor}
\end{align}
Assumptions \ref{as:densities-alt} and \ref{as:mle} yield
\begin{align*}
\sqrt{n}\E\iint \left|\hat f^w (z)-f^w(z)\right|\,dz &\le \sqrt{n} \E\left[ \iint \left\|\frac{\partial f^w}{\partial \gamma_w'}(z; \overline \gamma_w)\right\| \|\hat\gamma_w-\gamma_w\| dz  \right] \\
&\le C \sqrt{n} \E \|\hat\gamma_w-\gamma_w\| = O(1).    
\end{align*}
Therefore,
\begin{align*}
\E\left[ \sup_{\|k-k_0\|_{\infty}< \delta}\sqrt n \left| (\M_n(k)-\M(k))-(\M_n(k_0)-\M(k_0))\right| \right]  \le O(1) \delta.
\end{align*}
Choosing $\phi_n (\delta)=\delta$ in Theorem 3.2.5 of \citet{vdv} establishes \eqref{eq:k-hat-consistent}.

Next, we show that
\begin{align}
      \left\|\hat f-f \right\|_1=O_p\left (1/\sqrt n\right ). \label{eq:f-hat-consistency}
\end{align}
Write
\begin{align}
    \hat f-f = e^{\hat k} \hat f^w - e^{k_0} f^w = e^{\hat k}(\hat f^w-f^w) + f^w(e^{\hat k}-e^{k_0}). \label{eq:f-hat-f}
\end{align}
The first term on the right can be bounded as 
\[
\left\| e^{\hat k}(\hat f^w-f^w) \right\|_1 \le e^{\sup_{k\in\mathcal{K} }\|k\|_\infty}  \|\hat f^w-f^w\|_1 = O_p(1/\sqrt{n}),
\]
where the rate follows by the Taylor expansion \eqref{eq:fw-hat-taylor} and Assumptions \ref{as:densities-alt} and \ref{as:mle}.
The second term can be bounded as 
\[
\left\|f^w(e^{\hat k}-e^k) \right \|_1 \le \|f^w\|_1 \|e^{\hat k}-e^k \|_{\infty} \le e^{\sup_{k \in \mathcal{K}} \|k\|_{\infty}}\|\hat k - k_0\|_\infty= O_p(1/\sqrt n),
\]
where the second inequality uses the bound $|e^x-e^y|\le e^{\max(x,y)} |x-y|$ due to the mean value theorem.
Combining the two bounds yields \eqref{eq:f-hat-consistency}.

Finally, let us show that $\hat\theta-\theta_0 =O_p\left (1/ \sqrt n \right )$. The moment estimator of $\theta$ is
the solution to
\[
\int \varphi(z,\hat \theta) \hat f(z) d z=0 
\]
Because
\[
\left [\int \varphi(z,\theta) f(z) d z \right ]' \left [\int \varphi(z,\theta) f(z) d z \right ]
\]
is uniquely minimized at $\theta=\theta_0$, and by
\begin{align*}
    \sup_{\theta \in \Theta} \left| \int \varphi(z,\theta)(\hat f(z)-f(z))\,dz \right| \le \sup_{\theta \in \Theta} \sup_{z \in \mathcal{Z}}|\varphi(z,\theta)| \cdot \|\hat f-f\|_1 = o_p(1)
\end{align*}
Since the function $\theta \mapsto \int \varphi(z, \theta) \hat f(z) d z$ converges uniformly in probability, we conclude by Theorem 2.1 in \citet{newey1994large} that $\hat \theta \stackrel{p}{\rightarrow}\theta_0$.

By the mean value theorem
\[
\int \left [ \varphi(z,\theta_0)+  \frac{\partial \varphi}{\partial \theta'}(z,\bar \theta)(\hat \theta -\theta_0) \right ](\hat f(z) -f(z))dz+\int \left [ \varphi(z,\theta_0)+  \frac{\partial \varphi}{\partial \theta'}(z,\bar \theta )(\hat \theta -\theta_0)\right ]f(z)dz=0.
\]
By the consistency of $\hat\theta$, the $\sqrt{n}$-consistency of $\hat f$ in equation \eqref{eq:f-hat-consistency}, and Assumption \ref{as:identification}(v), the first integral above is $O_p(1/\sqrt{n})$.
Therefore,
\[
 \left [\int   \frac{\partial \varphi}{\partial \theta'}(z,\theta_0) f(z)dz \right ](\hat \theta -\theta_0)= - \int \left [ \varphi(z,\theta_0)\right ](\hat f(z) -f(z))dz+O_p(1/\sqrt{n}).
\]
By a similar argument, the integral on the right is $O_p(1/\sqrt{n})$, and hence
\[
\hat \theta -\theta_0=-\left [\int   \frac{\partial \varphi}{\partial \theta'}(z,\theta_0) f(z)dz \right ]^{-1}\int  \varphi(z,\theta_0)(\hat f(z) -f(z))dz+O_p \left ( 1/\sqrt n \right )= O_p \left ( 1/\sqrt n \right ).
\]

\section{Variance of the raking estimator of $\theta=\E(Z_1Z_2)$ under the Gaussian DGP} \label{app:MC_var_theory}

\subsection{Theoretical variance}

Let $n$ be the number of period-1 observations, let $W_i$ indicate that unit $i$
is retained in period 2, and let $n_w=\sum_{i=1}^n W_i$ be the balanced-panel
sample size. The period-2 refreshment sample is independent of the rest of the data and has size $n_r$.

The two target marginal working models are
\[
 Z_1\sim N(\mu_1,a),\qquad Z_2^r\sim N(\mu_2,b),
\]
and the balanced-panel working model is
\[
 (Z_1,Z_2)'\mid W=1\sim
 N\!\left[
 \begin{pmatrix}\nu_1\\ \nu_2\end{pmatrix},
 C_w\right],
 \qquad
 C_w=\begin{pmatrix}u&c\\c&v\end{pmatrix}.
\]
The balanced-panel distribution need not actually be Gaussian. In that case,
$(\nu_1,\nu_2,u,v,c)$ are the pseudo-true Gaussian parameters, namely its
conditional means, variances, and covariance.

The raked Gaussian distribution has mean and covariance
\[
 \bm\mu=\begin{pmatrix}\mu_1\\ \mu_2\end{pmatrix},\qquad
 S=\begin{pmatrix}a&\tau\\ \tau&b\end{pmatrix}.
\]
The balanced-panel means $\nu_1,\nu_2$ do not enter the raked distribution.

Raking preserves the cross block of the precision matrix. To see this, note that a raking step for the first marginal multiplies the current density by a function of $z_1$ alone, while a step for the second marginal multiplies it by a function of $z_2$ alone. In the Gaussian log density, neither operation changes the cross-quadratic term of $z_1$ and $z_2$. Thus every raking iterate, and hence its limit, preserves the off-diagonal block of the precision matrix. Define
\[
 \Delta_w=uv-c^2,
 \qquad
 q=(C_w^{-1})_{12}=-\frac{c}{\Delta_w}.
\]
The cross block of $S^{-1}$ is $-\tau/(ab-\tau^2)$, so $\tau$ satisfies
\[
 q=-\frac{\tau}{ab-\tau^2}.
\]
The admissible solution is
\begin{equation}
 \tau=
 \frac{1-\sqrt{1+4q^2ab}}{2q}
 =-\frac{2qab}{1+\sqrt{1+4q^2ab}}
 \label{eq:tau}
\end{equation}
with $\tau=0$ when $q=0$.

Because the simulation parameter is
\[
 \theta=\E(Z_1Z_2)=\mu_1\mu_2+\tau,
\]
its plug-in estimator is
\begin{equation}
 \widehat\theta=\widehat\mu_1\widehat\mu_2+\widehat\tau,
 \label{eq:theta_hat}
\end{equation}
where $\widehat\tau$ is obtained from \eqref{eq:tau} using
$\widehat a,\widehat b,\widehat u,\widehat v,\widehat c$.

Let
\[
 D=\sqrt{1+4q^2ab}=1-2q\tau.
\]
Implicit differentiation of $q\tau^2-\tau-qab=0$ gives
\[
 \tau_a=-\frac{qb}{D},\qquad
 \tau_b=-\frac{qa}{D},\qquad
 \tau_q=\frac{\tau^2-ab}{D}.
\]
Moreover,
\[
 q_u=\frac{cv}{\Delta_w^2},\qquad
 q_v=\frac{cu}{\Delta_w^2},\qquad
 q_c=-\frac{uv+c^2}{\Delta_w^2},
\]
and therefore
\[
 \tau_u=\tau_q q_u,\qquad
 \tau_v=\tau_q q_v,\qquad
 \tau_c=\tau_q q_c.
\]
For
\[
 \gamma=(\mu_1,\mu_2,a,b,u,v,c)',
\]
the gradient is
\begin{equation}
 \nabla_\gamma\theta=
 (\mu_2,\mu_1,\tau_a,\tau_b,\tau_u,\tau_v,\tau_c)'.
 \label{eq:gradient}
\end{equation}

Under marginal Gaussian models,
\[
 Var(\widehat\mu_1)\simeq \frac{a}{n},\quad
 Var(\widehat a)\simeq \frac{2a^2}{n},\quad
 Cov(\widehat\mu_1,\widehat a)=0,
\]
and, for the independent refreshment sample,
\[
 Var(\widehat\mu_2)\simeq \frac{b}{n_r},\quad
 Var(\widehat b)\simeq \frac{2b^2}{n_r},\quad
 Cov(\widehat\mu_2,\widehat b)=0.
\]

For the balanced-panel covariance estimators,
\[
 Var\!\begin{pmatrix}\widehat u\\ \widehat v\\ \widehat c\end{pmatrix}
 \simeq \frac{1}{n_w}M_w,
\]
where
\begin{equation}
 M_w=
 \begin{pmatrix}
 2u^2&2c^2&2uc\\
 2c^2&2v^2&2vc\\
 2uc&2vc&uv+c^2
 \end{pmatrix}.
 \label{eq:Mw}
\end{equation}

The full period-1 sample and balanced panel overlap. If $p=\Pr(W=1)$,
the influence functions for $a$ and $C_w$ are, respectively,
\[
 (Z_1-\mu_1)^2-a
 \quad\text{and}\quad
 \frac{W}{p}
 \left\{(Z-\nu)(Z-\nu)'-C_w\right\}.
\]
Using conditional Gaussian fourth moments gives the nonzero overlap covariance
\begin{equation}
 Cov\!\left(
 \widehat a,
 \begin{pmatrix}\widehat u\\ \widehat v\\ \widehat c\end{pmatrix}
 \right)
 \simeq
 \frac{1}{n}
 \begin{pmatrix}2u^2\\2c^2\\2uc\end{pmatrix}.
 \label{eq:overlapcov}
\end{equation}
The denominator is $n$, not $n_w$. Conditional Gaussian third central moments
also imply
\[
 Cov\!\left(\widehat\mu_1,\widehat u,\widehat v,\widehat c\right)=0.
\]

Equivalently, for $\lambda=(a,b,u,v,c)'$,
\begin{equation}
 V_\lambda=
 \begin{pmatrix}
 \frac{2a^2}{n}&0&\frac{2u^2}{n}&\frac{2c^2}{n}&\frac{2uc}{n}\\
 0&\frac{2b^2}{n_r}&0&0&0\\
 \frac{2u^2}{n}&0&\frac{2u^2}{n_w}&\frac{2c^2}{n_w}&\frac{2uc}{n_w}\\
 \frac{2c^2}{n}&0&\frac{2c^2}{n_w}&\frac{2v^2}{n_w}&\frac{2vc}{n_w}\\
 \frac{2uc}{n}&0&\frac{2uc}{n_w}&\frac{2vc}{n_w}&\frac{uv+c^2}{n_w}
 \end{pmatrix}.
 \label{eq:Vlambda}
\end{equation}

Combining \eqref{eq:gradient}--\eqref{eq:Vlambda} yields
\begin{align}
 Var(\widehat\theta)\simeq{}&
 \frac{\mu_2^2a}{n}+\frac{\mu_1^2b}{n_r}
 +\frac{2a^2\tau_a^2}{n}+\frac{2b^2\tau_b^2}{n_r}
 \notag\\
 &+\frac{1}{n_w}
 \begin{pmatrix}\tau_u&\tau_v&\tau_c\end{pmatrix}
 M_w
 \begin{pmatrix}\tau_u\\\tau_v\\\tau_c\end{pmatrix}
 \notag\\
 &+\frac{4\tau_a}{n}
 \left(u^2\tau_u+c^2\tau_v+uc\tau_c\right).
 \label{eq:mainvar}
\end{align}
The last line is the overlap correction. Written without matrix notation, the
balanced-panel quadratic form is
\begin{align*}
 &\frac{1}{n_w}\bigl[
 2u^2\tau_u^2+2v^2\tau_v^2+(uv+c^2)\tau_c^2
 +4c^2\tau_u\tau_v
 +4uc\tau_u\tau_c+4vc\tau_v\tau_c
 \bigr].
\end{align*}

\subsection{Numerical evaluation}

In the simulation design,
\[
 (Z_1,Z_2)'\sim N\!\left[
 \begin{pmatrix}0\\0\end{pmatrix},
 \begin{pmatrix}1&0.4\\0.4&1\end{pmatrix}
 \right]
\]
and the selection probability is
\[
 \Pr(W=1\mid Z_1,Z_2)
 =\exp\{-0.1|Z_1|-0.3|Z_2|\}.
\]
We have $N=n+n_r$, with
\[
 n=0.6N,\qquad n_r=0.4N.
\]
Numerical integration of
\[
 E\{g(Z_1,Z_2)\mid W=1\}
 =\frac{E[g(Z_1,Z_2)e^{-0.1|Z_1|-0.3|Z_2|}]}
 {E[e^{-0.1|Z_1|-0.3|Z_2|}]}
\]
gives
\[
 p=0.7403301,\qquad n_w\simeq np=0.4441981N,
\]
and, by symmetry, $\nu_1=\nu_2=0$. The pseudo-true balanced-panel covariance is
\begin{equation}
 C_w=
 \begin{pmatrix}
 0.8938931&0.2922627\\
 0.2922627&0.7819245
 \end{pmatrix}.
 \label{eq:Cwnum}
\end{equation}
Thus,
\[
 \mu_1=\mu_2=0,\qquad a=b=1,
\]
\[
 \Delta_w=0.6135395,\qquad q=-0.4763552,\qquad D=1.3811796,
\]
and
\[
 \tau=\theta_{\mathrm{GR}}=0.4001002.
\]
This is extremely close to the true
target $\theta_0=0.4$, even though the balanced-panel Gaussian model is
misspecified.

The numerical derivatives are
\[
 \tau_a=\tau_b=0.3448901,\qquad \tau_q=-0.6081178,
\]
\[
 q_u=0.6070902,\qquad q_v=0.6940232,\qquad q_c=-2.0837156,
\]
and
\begin{equation}
 \tau_u=-0.3691824,\qquad
 \tau_v=-0.4220478,\qquad
 \tau_c=1.2671445.
\end{equation}
Consequently,
\[
 \nabla_\gamma\theta=
 (0,0,0.3448901,0.3448901,-0.3691824,-0.4220478,1.2671445)'.
\]
The marginal means make no first-order contribution because $\mu_1=\mu_2=0$.

The numerical balanced-panel moment matrix is
\begin{equation}
 M_w=
 \begin{pmatrix}
 1.5980897&0.1708350&0.5225032\\
 0.1708350&1.2228120&0.4570548\\
 0.5225032&0.4570548&0.7843744
 \end{pmatrix}.
\end{equation}

For these parameter values,
\[
 u^2\tau_u+c^2\tau_v+uc\tau_c=0
\]
up to numerical precision. Hence the overlap contribution in \eqref{eq:mainvar}
is exactly zero in this simulation. This does \emph{not} mean that the input
estimators are independent: \eqref{eq:overlapcov} remains nonzero. Rather, the
gradient of $\tau$ is orthogonal to the relevant overlap-covariance direction.

The two remaining components of the variance are
\[
 \frac{2\tau_a^2}{0.6N}+\frac{2\tau_b^2}{0.4N}
 =\frac{0.9912433}{N}
\]
and
\[
 \frac{1}{0.4441981N}
 (\tau_u,\tau_v,\tau_c)M_w(\tau_u,\tau_v,\tau_c)'
 =\frac{1.7347510}{N}.
\]
Therefore,
\begin{equation}
 Var(\widehat\theta)\simeq\frac{2.7259943}{N},
 \qquad
 \operatorname{se}(\widehat\theta)\simeq\frac{1.6510585}{\sqrt N}.
 \label{eq:finalnum}
\end{equation}
For the sample sizes reported in the simulation table:
\[
\begin{array}{@{}ccc@{}}
\toprule
N & \text{Asymptotic variance} & \text{Asymptotic standard error}\\
\midrule
1000 & 0.0027260 & 0.05220\\
5000 & 0.0005452 & 0.02335\\
\bottomrule
\end{array}
\]
The Monte Carlo standard errors reported in section \ref{sec:mc} are approximately
$0.054$ and $0.029$, respectively, which are very close to the theoretical values above.

\section{Definitions of empirical variables}\label{app:variables}

\textbf{Numeracy}
A respondent’s numeracy score is based on answers to the following 8 questions [with the answers accepted as correct given in square brackets]. The question codes correspond to the UAS documentation: \url{https://uasdata.usc.edu/page/Cognitive+Comprehensive+File}.

\begin{description}
    \item[lip001] Imagine that we roll a fair, six-sided die 1,000 times. Out of 1,000 rolls, how many times do you think the die would come up as an even number? [490-510]
    \item[lip002] In the BIG BUCKS LOTTERY, the chances of winning a \$10.00 prize are 1\%. What is your best guess about how many people would win a \$10.00 prize if 1,000 people each buy a single ticket from BIG BUCKS? [10]
    \item[lip003] In the ACME PUBLISHING SWEEPSTAKES, the chance of winning a car is 1 in 1,000.
What percent of tickets of ACME PUBLISHING SWEEPSTAKES win a car? [0.1]
    \item[lip008] If the chance of getting a disease is 10%, how many people would be expected to get the disease out of 1000? [100]
    \item[lip009] If the chance of getting a disease is 20 out of 100, this would be the same as having what percent chance of getting the disease? [20]
    \item[lip012] Suppose you have a close friend who has a lump in her breast and must have a
mammogram. Of 100 women like her, 10 of them actually have a malignant tumor, and 90
of them do not. Of the 10 women who actually have a tumor, the mammogram indicates
correctly that 9 of them have a tumor, and indicates incorrectly that 1 of them does not
have a tumor. Of the 90 women who do not have a tumor, the mammogram indicates
correctly that 80 of them do not have a tumor, and indicates incorrectly that 10 of them do
have a tumor. The table below summarizes all of this information. \\
\small{
\begin{tabular}{|c|c|c|c|}
\hline
  & Tested Positive & Tested Negative & Totals \\ \hline
Actually has a tumor & 9 & 1 & 10 \\ \hline
Does not have a tumor & 10 & 80 & 90 \\ \hline
Totals & 19 & 81 & 100 \\ \hline
\end{tabular}    
} \\
Imagine that your friend tests positive (as if she had a tumor). What is the
likelihood that she actually has a tumor? [9 out of 19]
    \item[lip015] A bat and a ball cost \$1.10 in total. The bat costs \$1.00 more than the ball. How much does the ball cost? [5 or .05]
    \item[lip017] In a lake, there is a patch of lily pads. Every day, the patch doubles in size. If it takes 48 days for the patch to cover the entire lake, how long would it take for the patch to cover half of the lake? [47]
\end{description}

The numeracy scores were derived using a two-parameter logistic IRT model. Scores are normalized to a mean of 50 and a standard deviation of 10. Higher scores indicate better performance on numeracy problems. \\

\textbf{Diabetes}
This variable equals one if the respondent reports having diabetes and zero otherwise. \\

\textbf{College}
This variable equals one if the respondent reports having at least a bachelor’s degree and zero otherwise. \\

\textbf{Depressed}
This variable equals one if the respondent confirms at least three symptoms on the CES-D 8 questionnaire (Items 4 and 6 are reverse-scored, i.e., if a respondent denies such an item, that adds to the score). The eight CES-D statements that respondents could confirm or deny are:
\begin{enumerate}
    \item Much of the time during the past week, you felt depressed.
    \item You felt that everything you did was an effort.
    \item Your sleep was restless.
    \item You were happy.
    \item You felt lonely.
    \item You enjoyed life.
    \item You felt sad.
    \item You could not get going.
\end{enumerate}